\documentclass[preprint,11pt,authoryear]{elsarticle}

\usepackage{amsmath,amssymb,amsthm}
\usepackage{xcolor}

\usepackage{booktabs}
\usepackage{multirow}

\usepackage{url}

\usepackage{times}

\def\iid{\buildrel {\rm i.i.d.} \over \sim}
\def\eqd{\buildrel {\rm d} \over =}

\def\Real{\mathbb{R}}
\def\P#1{{\mathbb{P}}\left(#1\right)}
\def\E#1{{\mathbb E}\left[#1\right]}
\def\Var#1{{\mathrm Var}\left(#1\right)}

\newtheorem{example}{Example}
\newtheorem{model}{Model}
\newtheorem{definition}{Definition}
\newtheorem{proposition}{Proposition}
\newtheorem{theorem}{Theorem}

\newtheorem{lemma}{Lemma}
\newtheorem{remark}{Remark}

\renewcommand{\emph}[1]{\textbf{#1}}

\begin{document}
\begin{frontmatter}
\title{Can a small additional claim lower the premium? Credibility orders for collective risk models}
\author[EH]{Jihyun Park \corref{cor1}}

\author[EH]{Jieun Kim \corref{cor1}}

\author[UR]{Taehan Bae \corref{cor1}}

\author[EH]{Jae Youn Ahn \corref{cor2}}
\ead{jaeyahn@ewha.ac.kr}

\address[EH]{Department of Statistics, Ewha Womans University, Seoul, Republic of Korea}

\address[UR]{Department of Mathematics, University of Regina, Regina, Canada}

\cortext[cor1]{First Authors}
\cortext[cor2]{Corresponding Author}

\begin{abstract}
The collective risk model is a fundamental framework in insurance ratemaking for modeling aggregate losses by combining claim frequency and claim severity components. 
A key structural requirement for a reliable experience rating system is a monotone ordering property: policyholders with worse past experience should receive a stochastically larger prediction for future losses, and hence a higher premium. 
Such credibility-type monotonicity is well documented in the insurance literature for univariate outcomes under classical random-effect models. 
However, extending this principle to the collective risk model is nontrivial because the relevant history is inherently multivariate, involving both claim counts and individual claim amounts. Hence, despite its practical importance, a corresponding credibility-type ordering has not been systematically developed for predictive distributions of aggregate loss.
In this paper, we provide an example showing that standard collective risk model specifications can violate monotonicity: adding an additional but sufficiently small claim may decrease the premium, thereby creating perverse incentives for strategic reporting and potentially undermining the integrity of experience rating.
Motivated by this pathology in view of insurance, we formalize a credibility order tailored to collective risk models and derive tractable sufficient conditions under which the predictive distribution of aggregate loss is monotone in past experience, ruling out such pathological violations.
Numerical studies and an empirical illustration using real insurance data accompany our theoretical results.
\end{abstract}


\end{frontmatter}

\vfill

\pagebreak

\vfill

\pagebreak

\section{Introduction}

A basic requirement of a well-functioning experience rating system is
\emph{stochastic monotonicity} of the predictive distribution, which plays a central role in risk assessment and pricing models in the area of economics, operations research, insurance as well as statistics \citep{Lehmann1988, milgrom1994monotone, athey2018value, agoston2026stationary}; specifically, as a policyholder reports more claims (or larger claims), the insurer's
assessment of that policyholder's future risk should not become smaller.
This principle is more than a modeling preference, and it is a structural requirement for a reliable
ratemaking system, because it ensures that premiums respond in the correct direction as the
observed claim history deteriorates.
In the classical random-effect framework for univariate outcomes, such monotonicity can be
formalized using stochastic orders of posterior, and has been
extensively developed under the name of 
the \emph{conditionally increasing in sequence} in the actuarial literature; see, for example,
\citet{purcaru2003dependence, denuit2006actuarial, agoston2026stationary, ahn2026zero}.
Following \citet{ahn2026zero}, we refer to this stochastic-order formulation  as a \emph{credibility order} to emphasize the
insurance perspective.

Many insurance products, however, are most naturally priced not through a single univariate
outcome but through the \emph{aggregate loss} over a period.
This leads to the \emph{collective risk model} (CRM), a workhorse framework in which claim
frequency and claim severity are modeled separately and then aggregated. 
Concretely, letting $N_t$ be the number of claims reported in period $t$ and
\begin{equation}\label{eq.r.1}
\boldsymbol{Y}_{t}:=
\begin{cases}
(Y_{t, 1}, \cdots, Y_{t, N_t}), & N_t>0;\\
\emptyset, & N_t=0.
\end{cases}
\end{equation}
the corresponding severities, the aggregate loss is
\[
S_t := \sum_{j=1}^{N_t} Y_{t,j}
\]
with the convention $S_t=0$ when $N_t=0$.
For ratemaking purposes, collective risk models are often embedded in (dynamic) random-effect
structures to capture persistent heterogeneity across policyholders and serial dependence over time;
see, for example, \citet{pinquet1997allowance, boucher2007risk, hernandez2009net, shi2014multivariate, baumgartner2015bayesian, denuit2021wishart, verschuren2022frequency, jeong2023multivariate, youn2023simple}.

Despite their practical importance, credibility-type monotonicity for aggregate loss
processes have received far less attention than their univariate counterparts.
In particular, credibility-type monotonicity results are typically formulated for predictive
distributions based on \emph{coarsened} summaries of the past \citep{Frees01012003, purcaru2003dependence}, such as
$S_{t+1}\mid S_{1}, \cdots, S_t$,  and are widely used for forecasting and risk evaluation
\citep{hewitt1970credibility, denuit2006actuarial, goulet2006credibility}.
By contrast, the corresponding ordering problem for the predictive distribution of $S_{t+1}$ conditional on the \emph{full-claim history} up to time $t$, 
\[
S_{t+1}\mid N_{1},\cdots, N_t, \boldsymbol{Y}_{1}, \cdots, \boldsymbol{Y}_{t},
\]
has not been systematically analyzed in the litarature.
This gap matters because the full-claim history contains the most detailed information about past
experience in the collective risk model. It is therefore the natural conditioning information for
checking whether a worse claim history could (paradoxically) lead to a smaller predicted aggregate
loss.

Related work has largely focused on the informational value of different summaries of past
experience.
For example, \citet{goulet2006credibility} and \citet{oh2021predictive} compare the efficiency of posterior risk analysis when one
uses past frequency information versus past aggregate loss information in a collective risk model.
Nevertheless, the core credibility-type order remains open in that line of work.

Indeed, a key motivation for this paper is a subtle but practically serious pathology.
Under classical collective risk model specifications in which frequency and severity are driven by two distinct latent variables, it is possible for a policyholder to worsen the observed history in a natural
sense---for instance, by reporting an additional claim of very small size---and yet obtain a
smaller predictive distribution (hence a smaller predictive mean) for the next aggregate
loss.
Such a mechanism creates a perverse incentive: reporting an extra (small) accident can reduce the
future premium.
Incentives of this kind are undesirable because they can distort reporting behavior, undermine the
integrity of experience rating, and ultimately threaten the stability of the insurance pool \citep{pauly1968economics, arrow1978uncertainty,   boyer1989empirical, Dionne1991}.

The purpose of this paper is to develop credibility-type ordering results for the collective risk
model that rule out such perverse incentives.
Our contributions are twofold.
First, we give an illustrative example demonstrating that standard collective risk models can fail
to satisfy natural monotonicity requirements: an additional claim of sufficiently small size may
decrease the predictive distribution, and hence the premium, for future aggregate loss.
Second, motivated by this pathology, we formalize an appropriate credibility-order notion for the
collective risk model by introducing a partial order on the space of full-claim histories, and we
derive tractable sufficient conditions under which worse histories lead to stochastically larger
predictive distributions for the next aggregate loss.

The remainder of the paper is organized as follows.
Section \ref{sec.2} reviews the necessary background on stochastic orders and notation, and then Section \ref{sec.3} presents a
motivating example demonstrating how a standard collective risk model can create non-monotone premium responses and
strategic incentives.
Section~\ref{sec.4} introduces the full-claim ordering framework and establishes sufficient
conditions for the full-claim credibility order under comonotonic assumptions on latent
variables and additional constraints on the link functions. 
Then, the numerical study is given in Section \ref{sec.5}, and the extension of the results in Section \ref{sec.4} is presented in Section \ref{sec.6}
Technical proofs and auxiliary lemmas are collected in the Appendix.

\section{Preliminaries on stochastic orders}\label{sec.2}

\subsection{Symbols and definitions}
We first introduce basic notation and conventions used throughout the paper. 
The sets $\Real$, $\Real_{\ge 0}$, $\Real_{>0}$, $\mathbb{N}_{0}$, and
$\mathbb{N}$ denote, respectively, the real numbers, the nonnegative real
numbers, the strictly positive real numbers, the nonnegative integers, and
the strictly positive integers.
For $t\in\mathbb{N}$, let $\Real^{t}$ to denote the $t$-dimensional Euclidean space,
with $\Real_{\ge 0}^{\,t}$ and $\Real_{>0}^{\,t}$ denoting its nonnegative and
strictly positive orthants; analogous notation is used for
$\mathbb{N}_{0}^{\,t}$ and $\mathbb{N}^{\,t}$.

We write
\[
y_{1:t} := (y_1,\ldots,y_t)\in\mathbb{R}^t,
\]
with the convention that $y_{1:0}$ denotes the empty sequence.
For vectors $x_{1:t}, y_{1:t} \in \mathbb{R}^t$, let
$x_{1:t} \wedge y_{1:t}$ and $x_{1:t} \vee y_{1:t}$ denote the vectors of their componentwise minimum and maximum, respectively.
We write $x_{1:t} \le y_{1:t}$ if the inequality holds componentwise.
We use the notation
\[
\langle x_{1:t}, y_{1:t}\rangle
\]
to denote the inner product between $x_{1:t}, y_{1:t}\in\mathbb{R}^t$. 
For two random variables \(X\) and \(Y\), we write \(X \perp Y\) to indicate that they are independent.

We next specify the parametric distributions that will be repeatedly employed
in later sections.

\begin{itemize}
  \item \textbf{Gamma distribution.}
  A random variable $X$ is said to follow a ${\rm Gamma}(\alpha,\theta)$
  distribution with shape parameter $\alpha>0$ and scale parameter
  $\theta>0$ if its probability density function is given by
  \[
  f(x)
  =
  \begin{cases}
  \displaystyle
  \frac{1}{\Gamma(\alpha)\,\theta^{\alpha}}
  x^{\alpha-1}\exp\!\left(-\frac{x}{\theta}\right),
  & x>0,\\[1.2ex]
  0, & \text{otherwise}.
  \end{cases}
  \]
  In this parameterization, the mean and variance are
  \[
  \mathbb{E}(X)=\alpha\theta,
  \qquad
  \mathrm{Var}(X)=\alpha\theta^{2}.
  \]
 \item \textbf{Inverse\text{-}Gamma}.
     A random variable $X$ is said to follow\(\mathrm{Inv\text{-}Gamma}(\kappa,\theta)\) if its probability density function is given by
\[
f(x)=
\begin{cases}
  \frac{\theta^\kappa}{\Gamma(\kappa)}x^{-\kappa-1}\exp(-\theta/x), & x>0  \\
  0, & \mbox{otherwise}.
\end{cases}
\]
In this parametrization, \(\E{X}=\theta/(\kappa-1)\) for \(\kappa>1\), and
$\Var{X}=\frac{\theta^2}{(\kappa-1)^2(\kappa - 2 )}$ for $\kappa>2$.
If \(X\sim \mathrm{Gamma}(\alpha,\theta)\), where \(\theta\)
is the scale parameter, then
\[
    \frac{1}{X}\sim \mathrm{Inv\text{-}Gamma}\left(\alpha,\frac{1}{\theta}\right).
\]

  \item \textbf{Weibull distribution.}
  A random variable $X$ is said to follow an ${\rm Weibull}(\alpha,\beta)$ with shape parameter $\alpha>0$ and scale parameter $\beta>0$ if its probability density function is given by
\[
f(x)
=
\begin{cases}
\displaystyle
\frac{\alpha}{\beta} \left( \frac{x}{\beta}\right)^{\alpha-1}{\rm exp}\left( -(x/\beta)^\alpha\right)
,
& x>0,\\[1.2ex]
0, & \text{otherwise}.
\end{cases}
\]

  \item \textbf{Transformed Gamma distribution.}
  A random variable $X$ is said to follow an ${\rm Transformed\text{-}Gamma}(\alpha,\tau, \beta)$ with the first shape parameter $\alpha>0$, the second shape parameter $\tau>0$ and scale parameter $\beta>0$ \citep{kleiber2003statistical, klugman2012loss} if its probability density function is given by
\[
f(x)
=
\begin{cases}
\displaystyle
\frac{\tau}{x\Gamma(\alpha)} \left( \frac{x}{\beta}\right)^{\alpha\tau}
{\rm exp}\left( - (x/\beta)^\tau\right)
,
& x>0,\\[1.2ex]
0, & \text{otherwise}.
\end{cases}
\]
  \item \textbf{Poisson distribution.}
  A random variable $N$ is said to follow a ${\rm Pois}(\lambda)$ distribution
  with intensity parameter $\lambda>0$ if
  \[
  \mathbb{P}(N=n)=\frac{\lambda^{n}}{n!}\exp(-\lambda),
  \qquad n\in\mathbb{N}_{0}.
  \]
\end{itemize}

\subsection{Review: stochastic orders}


For univariate random variables $X$ and $Y$ with densities $f$ and $g$, 
we say that $X$ is smaller than $Y$ in the \emph{likelihood-ratio order (LR order)},
denoted $X \le_{\mathrm{LR}} Y$, if
\[
f(x)\, g(y) \;\ge\; f(y)\, g(x),
\qquad \hbox{for all}\quad x \le y.
\]
Equivalently, the likelihood ratio $f(x)/g(x)$ must be decreasing in $x$ which means that larger values of the variable are relatively more
likely under $Y$ than under $X$.  
The LR order is one of the strongest stochastic orders and will be used below
to characterise how predictive distributions respond to differences in past
observations.

Similarly, $X$ is said to be smaller than $Y$ in the \emph{stochastic order (ST order)},
denoted $X \le_{\mathrm{ST}} Y$, if
\[
\mathbb{P}(X > x) \;\le\; \mathbb{P}(Y > x),
\qquad \hbox{for all}\quad x\in\Real.
\]
It is well known that $X \le_{\mathrm{ST}} Y$ if and only if
$\E{\phi(X)} \le \E{\phi(Y)}$ for all bounded, nondecreasing functions $\phi$,
or, equivalently, for all nondecreasing functions $\phi$ for which the expectations are finite;
see, for example, \citet{denuit2006actuarial, shaked2007stochastic}.

The LR order implies the ST order \citep{denuit2006actuarial, shaked2007stochastic}, and the latter is closely related to premium calculation as shown in Definition \ref{def.1} below. We also have the well-known result on the ST order under the mixture distribution \citep{denuit2006actuarial, shaked2007stochastic}.


\begin{proposition}\citep{denuit2006actuarial, shaked2007stochastic}
\label{prop.p1}
Let $X$ and $\Theta$ be real-valued random variables. Fix a version of the
conditional distribution of $X$ given $\Theta=\theta$, and assume that
\[
\theta_1\le \theta_2 
\quad \Longrightarrow \quad
\left[X\mid \Theta=\theta_1\right]\le_{\mathrm{ST}}
\left[X\mid \Theta=\theta_2\right].
\]
Let $G_1$ and $G_2$ be two candidate distribution functions for the mixing
variable $\Theta$.  For each $i\in\{1,2\}$, let $X_i$ be a real-valued random variable whose
distribution function is given by
\[
\P{X_i\le x}
=
\int_{\Real} \P{X\le x \mid \Theta=\theta}\,dG_i(\theta),
\qquad x\in\Real.
\]
If the mixing distributions are stochastically ordered, that is,
\[
G_1(x)\ge G_2(x),
\qquad x\in\Real,
\]
then
\[
X_1 \le_{\mathrm{ST}} X_2.
\]
\end{proposition}

\subsection{Review: credibility orders and collective risk models}

A central requirement in experience rating is credibility-type monotonicity: policyholders with
\emph{worse} past experience should receive a stochastically larger prediction, and hence a
higher premium, for future losses. In the stochastic-order literature, this property is often
viewed as a form of \emph{conditional increasingness in sequence} \citep{denuit2006actuarial, shaked2007stochastic}. Following \citet{ahn2026zero}, we refer to this
property as the \emph{credibility order} to emphasize the insurance perspective. We now give
the formal definition.

\begin{definition}[Credibility order]\label{def.1}
We say that the process $(X_t)_{t\ge1}$ satisfies the credibility order
if, for every $t\ge1$ and all histories $x_{1:t}\le x'_{1:t}$,
\[
\E{\,h(X_{t+1})\mid X_{1:t}=x_{1:t}\,}
\ \le\
\E{\,h(X_{t+1})\mid X_{1:t}=x'_{1:t}\,}
\]
for every nondecreasing measurable function $h$ such that the expectation exists.
\end{definition}
Note that showing the credibility order is equivalent to showing that
\[
\left[X_{t+1} \mid X_{1:t}=x_{1:t}\right]
\ \le_{\mathrm{ST}}\
\left[X_{t+1} \mid X_{1:t}=x'_{1:t}\right]
\]
for every $t\ge1$ and all histories $x_{1:t}\le x'_{1:t}$.

Credibility order plays a central role in insurance ratemaking.
Within the (dynamic) random-effect model framework, such ordering properties
were first established for Poisson models in the seminal work of
\citet{purcaru2003dependence}, and later extended to models whose conditional distributions belong to the exponential dispersion family \citep{denuit2006actuarial}.
However, a corresponding credibility order for
the collective risk model has not yet been systematically studied in the
literature.
This paper is therefore primarily concerned with establishing credibility
orders within the framework of the collective risk model \citep{hernandez2009net, frees2014predictive, garrido2016generalized}. 
For clarity, Sections~\ref{sec.3} and~\ref{sec.4} focus on the
Poisson--Gamma collective risk model, where claim frequencies follow a
Poisson distribution and claim severities follow a Gamma distribution.
Extensions to broader distributional families are presented in
Section~\ref{sec.6}.

\begin{model}\label{mod.0}
Assume that the latent vector $(R_1,R_2)$ has a prior density function $\pi:\Real^2\mapsto \Real_{\ge 0}$. Let $\sigma_1$ and $\sigma_2$ be known non-decreasing activation functions, assumed to yield
valid Poisson intensities and Gamma scales in i and ii.
Conditionally on $(R_1, R_2)$ and known exposures 
\[
\{\lambda_t\}_{t\ge 1}, \quad \{\mu_t\}_{t\ge 1}\subset \Real_{>0}, \quad a>0,
\]
 assume the followings:
\begin{enumerate}
\item[i.] \textbf{Frequency:} $N_t \mid R_1, R_2 \sim \mathrm{Poisson}(\lambda_t \sigma_1(R_1))$ independently over $t$.

\item[ii.] \textbf{Severity:}
Conditional on \((R_1,R_2)\), the individual severities
\(\{Y_{t,j}\}_{t\ge1,\ j\ge1}\) are mutually independent, with
\[
\left[Y_{t,j}\mid R_1,R_2\right]
\sim
\mathrm{Gamma}\left(a,\mu_t\sigma_2(R_2)\right).
\]
where \(a>0\). Moreover, conditional on \((R_1,R_2)\), they are independent of
\(\{N_t\}_{t\ge1}\).

\item[iii.] \textbf{Aggregate loss:} the aggregate loss at time $t$ is
\[
S_t:=\sum_{j=1}^{N_t}Y_{t,j}
\]
with the convention $S_t=0$ when $N_t=0$.
\end{enumerate}
\end{model}

Recall we have $\boldsymbol{Y}_t$ defined in \eqref{eq.r.1} for given $N_t\in\mathbb{N}$.
We also have 
\[
N_{1:t}:=\left( N_1, \cdots, N_t\right)\quad\hbox{and}\quad 
\boldsymbol{Y}_{1:t}:= \left(\boldsymbol{Y}_{1}, \cdots, \boldsymbol{Y}_{t} \right).
\]
We define $n_{1:t}$,  $\boldsymbol{y}_{1:t}$ and $s_{1:t}$ as the realization of $N_{1:t}$, $\boldsymbol{Y}_{1:t}$ and $S_{1:t}$, respectively.
We  call
\[
\boldsymbol{h}_t:=\left(n_{1:t}, \boldsymbol{y}_{1:t}\right)
\]
the \emph{full-claim history} up to time $t$,
and let \(\mathcal{H}_t\) denote the set of all admissible full-claim-histories up to time \(t\);
\[
\mathcal{H}_t
:= \Bigl\{ (n_{1:t},\boldsymbol{y}_{1:t}) \,:\,
n_{1:t}\in \mathbb{N}_{0}^t,\ \boldsymbol{y}_j \in \mathbb{R}_{>0}^{\,n_j}\ \text{for all } j=1,\ldots,t
\Bigr\},
\]
with the convention \(\mathbb{R}_{>0}^{\,0}:=\{\emptyset\}\).

When writing the (conditional) density function of a random
quantity in Model~\ref{mod.0}, we use the generic symbol $f(\,\cdot\,)$ and
suppress subscripts. 
For example, for the claim counts, when conditioning on $N_{1:t}=n_{1:t}$ we
simply write
\[
f(n_{t+1}\mid n_{1:t})
:=\P{\,N_{t+1}=n_{t+1}\mid N_{1:t}=n_{1:t}\,},
\qquad n_{t+1}\in\mathbb{N}_{0}.
\]
We also use $\pi(r_1,r_2)$ to represent the prior density of $(R_1,R_2)$, and
the same notation applies to conditional distributions of $(R_1,R_2)$ such as
\[
\pi(r_1,r_2\mid \boldsymbol{h}_t),
\]
which denotes the posterior density of $(R_1,R_2)$ given the full-claim-history
$\boldsymbol{h}_t\in\mathcal{H}_t$. Likewise, we sometimes suppress the explicit reference to the underlying
random history when conditioning. For instance, instead of writing
\[
N_{t+1}\mid N_{1:t}=n_{1:t},
\]
we may simply write
\[
N_{t+1}\mid n_{1:t},
\]
whenever it is clear from the context that $n_{1:t}$ is the realized value of
$N_{1:t}$.

\section{Motivating example}\label{sec.3}

Although Model~\ref{mod.0} is mathematically well defined, it may be problematic
from a ratemaking perspective. In particular, even if the usual credibility
order holds separately for the frequency and severity components, as in
\eqref{eq.4} and \eqref{eq.5} below, this does not in general guarantee
monotonicity of the predictive distribution of the aggregate loss.
Indeed, we present an example in which worsening a claim history by adding an
additional (but small) claim decreases the predicted aggregate loss,
thereby creating a perverse incentive for strategic reporting. If such
incentives are present, they can undermine the integrity
of experience rating and, in turn, threaten the reliability of the entire ratemaking
system.
To highlight this phenomenon, we consider the following simple and concrete special case
of Model~\ref{mod.0}.

\begin{model}\label{ex.0}
  Consider the setting of Model \ref{mod.0} with 
  \begin{itemize}
    \item Identity activation functions: $\sigma_1(x)=\sigma_2(x)=x$ for all $x\in\Real$
    \item Let $\mu_t=1$  and
    \(
    \lambda_t=1
    \) for all $t\in\mathbb{N}$
    \item $R_1\perp R_2$ with 
    \[
      R_1 \sim \mathrm{Gamma}(\alpha,1/\alpha)
    \]
    and
    \[
    R_2 \sim \mathrm{Inv\text{-}Gamma}(\beta+1, \beta)
    \]
    for some $\alpha, \beta\in\Real_{>0}$. Here, note that the choice of parametrization gives
    \[
    \E{R_1} =\E{R_2} = 1.
    \]

  \end{itemize}
\end{model}

The prior distribution in Model~\ref{ex.0} provides conjugacy.
Simple calculations yield
\[
R_1 \perp R_2 \mid \boldsymbol{h}_t,
\]
with the following marginal posterior distributions:
\[
\left[R_1 \mid \boldsymbol{h}_t\right] \sim \mathrm{Gamma}\!\left(\alpha+\sum_{j=1}^t n_j,\,\frac{1}{\alpha+t}\right)
\]
and
\[
\left[R_2 \mid \boldsymbol{h}_t\right] \sim \mathrm{Inv\text{-}Gamma}\!\left(\beta+1+a\sum_{j=1}^t n_j,\,\beta+\sum_{j=1}^t s_j\right).
\]
These closed-form posteriors facilitate the analysis of the predictive distribution.

By the property of Gamma and Inv-Gamma distributions, see also \citet{purcaru2003dependence}, we have
\begin{equation}\label{eq.2}
 \left[R_1 \mid n_{1:t}\right] \le_{\operatorname{ST}} \left[R_1 \mid n_{1:t}'\right], \qquad 
 \hbox{for}\quad n_{1:t}\le n_{1:t}'.
\end{equation}
Furthermore, assuming $n_{1:t}=n_{1:t}'$, we also have
\begin{equation}\label{eq.3}
 \left[
 R_2 \mid n_{1:t}, \boldsymbol{y}_{1:t}
 \right] \le_{\operatorname{ST}} 
 \left[R_2 \mid n_{1:t}', \boldsymbol{y}_{1:t}'\right], \qquad 
 \hbox{for}\quad \boldsymbol{y}_{1:t}\le \boldsymbol{y}_{1:t}'.
\end{equation}

The results in \eqref{eq.2} and \eqref{eq.3} can further be used to derive the
credibility orders for the frequency and severity components, respectively.
First, for \(0<r_1\le r_1'\), the additive property of the Poisson distribution gives
\[
\left[ N_{t+1}\mid r_1\right]
\le_{\operatorname{ST}}
\left[ N_{t+1}\mid r_1'\right].
\]
Together with the posterior stochastic order in \eqref{eq.2}, Proposition~\ref{prop.p1}
therefore implies
\begin{equation}\label{eq.4}
\left[N_{t+1} \mid n_{1:t}\right] \le_{\operatorname{ST}}
\left[N_{t+1} \mid n_{1:t}'\right], \qquad 
 \hbox{for}\quad n_{1:t}\le n_{1:t}'.
\end{equation}

Similarly, assuming \(n_{1:t}=n_{1:t}'\), for \(0<r_2\le r_2'\) we have
\[
\left[ Y_{t+1,j}\mid r_2\right]
\le_{\operatorname{ST}}
\left[ Y_{t+1,j}\mid r_2'\right]
\]
by the multiplicative property of Gamma distribution.
Together with the posterior stochastic order in \eqref{eq.3}, Proposition~\ref{prop.p1}
gives
\begin{equation}\label{eq.5}
\left[Y_{t+1, j} \mid n_{1:t}, \boldsymbol{y}_{1:t}\right] \le_{\operatorname{ST}}
\left[Y_{t+1, j} \mid n_{1:t}', \boldsymbol{y}_{1:t}'\right], \qquad 
 \hbox{for}\quad \boldsymbol{y}_{1:t}\le \boldsymbol{y}_{1:t}'.
\end{equation}
The detailed derivations of such orders under the general random-effect model
can be found in \citet{purcaru2003dependence} and \citet{denuit2006actuarial}. See also Lemma \ref{lem.2} and \ref{lem.3} in \ref{app.2} for similar arguments.

Based on the stochastic orders in \eqref{eq.4} and \eqref{eq.5}, the model in Model \ref{ex.0} seemingly looks okay as a model for ratemaking in insurance. 
However, since our main concern is the prediction of the next aggregate loss \(s_{t+1}\), the orderings in
\eqref{eq.4} and/or \eqref{eq.5} are not sufficient to guarantee a safe ratemaking process.
The next example shows that the predictive mean of \(s_{t+1}\) can even decrease after a history is
worsened by adding an extra claim.
To formalize this perturbation, fix a history
\(\boldsymbol{h}_t=(n_{1:t},\boldsymbol{y}_{1:t})\in\mathcal{H}_t\). For \(\epsilon\in\Real_{> 0}\), define the modified history
\[
\boldsymbol{h}_t(\epsilon):=\left( n_{1:t}', \boldsymbol{y}_{1:t}'\right),
\]
where
\[
n_k' = n_k
\quad\text{and}\quad
\boldsymbol{y}_{k}' =\boldsymbol{y}_{k}
\qquad \text{for } k=1, \cdots, t-1,
\]
and
\[
n_t'=n_t+1
\quad\text{and}\quad
\boldsymbol{y}_t'=\begin{cases}
(y_{t,1}, \cdots, y_{t, n_t}, \epsilon), & n_t>0;\\
(\epsilon), & n_t=0.
\end{cases}
\]
That is, \(\boldsymbol{h}_t(\epsilon)\) is obtained from \(\boldsymbol{h}_t\) by adding one claim of size
\(\epsilon\) in period \(t\).

\begin{example}\label{ex.2}
Consider the setting in Model~\ref{ex.0}. Then
\[
\E{S_{t+1}\mid \boldsymbol{h}_t}
= a\,\E{R_1\mid \boldsymbol{h}_t}\,\E{R_2\mid \boldsymbol{h}_t}
= a\,
\frac{\alpha+\sum_{j=1}^t n_j}{\alpha+t}\cdot
\frac{\beta+\sum_{j=1}^t s_j}{\beta+a\sum_{j=1}^t n_j}.
\]
Moreover, we have
\[
\E{S_{t+1}\mid \boldsymbol{h}_t(\epsilon)}
= a\,
\frac{\alpha+\sum_{j=1}^t n_j+1}{\alpha+t}\cdot
\frac{\beta+\sum_{j=1}^t s_j+\epsilon}{\beta+a\left(\sum_{j=1}^t n_j+1\right)}.
\]

If 
\begin{equation}\label{eq.6.1}
a\alpha\le \beta,
\end{equation}
then we have
\begin{equation}\label{eq.7}
\E{S_{t+1}\mid \boldsymbol{h}_t}\le \E{S_{t+1}\mid \boldsymbol{h}_t(\epsilon)}
\end{equation}
which tells us that an additional claim surely increase(non-decrease) the premium.

The problematic case is when
\[
a\alpha>\beta.
\]
Then, for any \(\boldsymbol{h}_t\in\mathcal{H}_t\), the violation of the condition in \eqref{eq.6.1} results in 
\[
0<\epsilon<
\frac{\left(\beta+\sum_{j=1}^t s_j\right)(a\alpha-\beta)}
{\left(\alpha+\sum_{j=1}^t n_j+1\right)\left(\beta+a\sum_{j=1}^t n_j\right)}
\]
so that
\begin{equation}\label{eq.6}
\E{S_{t+1}\mid \boldsymbol{h}_t}> \E{S_{t+1}\mid \boldsymbol{h}_t(\epsilon)}.
\end{equation}
As a result, a policyholder with the current full-claim-history
$\boldsymbol{h}_t\in\mathcal{H}_t$ may have an incentive to file an additional
claim of small size $\epsilon>0$. This illustrates a
pathological feature of the model from a ratemaking perspective.


\end{example}

\begin{remark}\label{rem.1}
  While Model \ref{mod.0} with the condition in \eqref{eq.6.1} guarantees the order in \eqref{eq.7} of the base premium, it does not guarantee the order of other type of insurance, such as deductible insurance, in general. Indeed, 
  Table \ref{tab.1} shows such violations of the premium orders in case of the deductible insurance with the following parameter setting under the same model as in Model \ref{ex.0}:
\begin{itemize}
  \item parameter setting satisfying the condition in \eqref{eq.6.1}: $a=2$, $\alpha=1$, $\beta=3$
  \item claim history: $n_1=2$, $\boldsymbol{y}_1=(1.6, 1.6)$
  \item new claim amount: $\epsilon =0.1$
\end{itemize}
In all cases of $d=5, 10, 15, 20$ except for the case of $d=0$, we have
\[
\E{(S_{2}-d)_+ \mid \boldsymbol{h}_1} > \E{(S_{2}-d)_+ \mid \boldsymbol{h}_1(\epsilon)}.
\]
Again, this can be problematic, since it may
encourage unnecessary claim filings. 

\begin{table}[h!t!]
\centering
\begin{tabular}{c c c c c c}
\toprule
$d$
&
$\hbox{cred} =\E{(S_{2}-d)_+ \mid \boldsymbol{h}_1}$
&
&
$\hbox{cred}'=\E{(S_{2}-d)_+ \mid \boldsymbol{h}_1(\epsilon)}$
&
$\hbox{cred}-\hbox{cred}'$
&$\dfrac{\hbox{cred}}{\hbox{cred}'}$
\\
\midrule
$0$ & 2.657 (0.003) & $<$  & 2.800 (0.003) & -0.143 (0.005)   & 0.948   \\
$5$ & 0.629 (0.002) & $>$ & 0.549 (0.002) & 0.080 (0.003)  & 1.138   \\
$10$ & 0.172 (0.001) & $>$ & 0.117 (0.001) & 0.055 (0.002)  & 1.476   \\
$15$ & 0.057 (0.001) & $>$ & 0.029 ($\le$ 0.000) & 0.027 (0.001)  & 1.840  \\
$20$ & 0.021 (0.001) & $>$ & 0.009 ($\le$ 0.000) & 0.012 (0.001) & 2.376  \\
\bottomrule
\end{tabular}
\caption{Stop--loss premium comparison based on Monte Carlo simulation with \(1{,}000{,}000\) replications for different deductible levels \(d\). Monte Carlo standard errors are reported in parentheses.
}
\label{tab.1}
\end{table}

\end{remark}

\section{New credibility order for the collective risk model}\label{sec.4}

Example~\ref{ex.2} and Remark~\ref{rem.1} in Section \ref{sec.3} illustrate a potentially problematic feature of the
classical collective risk model from a ratemaking perspective.
Motivated by this issue, the goal of this section is to develop an analogue of the credibility order
in Definition~\ref{def.1} for the aggregate loss, namely,
\[
\left[S_{t+1}\mid \boldsymbol{h}_t\right] \le_{\operatorname{ST}}
\left[S_{t+1}\mid \boldsymbol{h}_t'\right],
\]
for histories \(\boldsymbol{h}_t,\boldsymbol{h}_t'\in\mathcal{H}_t\) such that
\(\boldsymbol{h}_t\) represents no better experience than \(\boldsymbol{h}_t'\).

However, this formulation is mathematically problematic as stated, because a componentwise
inequality between full-claim histories is not well defined in general: the dimensions of
\(\boldsymbol{h}_t\) and \(\boldsymbol{h}_t'\) may differ due to different claim counts.
Therefore, we first introduce an appropriate partial order on the space of full-claim histories
\(\mathcal{H}_t\). Based on this partial order, we then develop credibility-type ordering results for
aggregate losses that rule out such perverse incentives and support a reliable and well-behaved
ratemaking system.

\begin{definition}\label{def.2}
For \(\boldsymbol{h}_t=(n_{1:t},\boldsymbol{y}_{1:t})\in\mathcal{H}_t\) and
\(\boldsymbol{h}_t'=(n_{1:t}',\boldsymbol{y}_{1:t}')\in\mathcal{H}_t\), we define the \emph{full-claim order}
\(\boldsymbol{h}_t \le_{\mathcal{H}} \boldsymbol{h}_t'\) if the following conditions are satisfied
\begin{itemize}
  \item[i.] \(n_{1:t}\le n_{1:t}'\);
  \item[ii.] For each time point
\(j=1,\ldots,t\),
\[
y_{j,k}\le y'_{j,k},\qquad \text{for all positive integer}\quad k \le\min\{n_j,n_j'\}.
\]
\end{itemize}
\end{definition}

By Definition~\ref{def.2}, the relation $\le_{\mathcal{H}}$ defines a partial order \citep{davey2002introduction} on $\mathcal{H}_t$.
That is, for all $\boldsymbol{h}_t,\boldsymbol{h}_t',\boldsymbol{h}_t''\in\mathcal{H}_t$, we have:
\begin{itemize}
  \item \emph{Reflexivity:} $\boldsymbol{h}_t \le_{\mathcal{H}} \boldsymbol{h}_t$;
  \item \emph{Antisymmetry:} if $\boldsymbol{h}_t \le_{\mathcal{H}} \boldsymbol{h}_t'$ and $\boldsymbol{h}_t' \le_{\mathcal{H}} \boldsymbol{h}_t$, then $\boldsymbol{h}_t=\boldsymbol{h}_t'$;
  \item \emph{Transitivity:} if $\boldsymbol{h}_t \le_{\mathcal{H}} \boldsymbol{h}_t'$ and $\boldsymbol{h}_t' \le_{\mathcal{H}} \boldsymbol{h}_t''$, then $\boldsymbol{h}_t \le_{\mathcal{H}} \boldsymbol{h}_t''$.
\end{itemize}
With the partially ordered set $(\mathcal{H}_t,\le_{\mathcal{H}})$ in place, we can now define an
appropriate credibility order for the collective risk model.

\begin{definition}\label{def.4}
  Under the setting of the collective risk model, 
  we say that the model satisfies the \emph{full-claim credibility order} if the following condition is satisfied
   \[
   \left[S_{t+1}\mid \boldsymbol{h}_{t}\right] 
   \le_{\operatorname{ST}}
   \left[S_{t+1}\mid \boldsymbol{h}_{t}'\right]
   \]
   for every $t\ge 1$ and all histories $\boldsymbol{h}_{t}, \boldsymbol{h}_{t}'\in \mathcal{H}_t$ satisfying $\boldsymbol{h}_{t}\le_{\mathcal{H}} \boldsymbol{h}_{t}'$.

\end{definition}

The remainder of this section concentrates on the development of the collective risk model in Model \ref{mod.0} that satisfies the full-claim credibility order. Specifically, we consider the full-claim credibility order under a
comonotonic assumption between the two latent variables. 
As illustrated by Example~\ref{ex.2}, the difficulty arises because a worsening
of the full-claim history can send conflicting signals to the two random effects;
an additional small claim increases the posterior risk assessment for the random
effect in the frequency component, but the small claim size may simultaneously
lower the posterior risk assessment for the random effect in the severity
component. Thus, the frequency and severity effects may move in opposite
directions, making it unclear which effect dominates in the predictive
distribution of the aggregate loss.

The comonotonic random-effect specification resolves this conflict by replacing
the two-dimensional random effects with a single one-dimensional random effect
that drives both the frequency and severity components. Under this specification, the frequency and severity information updates the
predictive distribution through a single common random effect, rather than
through two separate random effects that may induce opposing stochastic-order
movements. Consequently, the overall effect of the observed history on
the predictive distribution becomes easier to identify and analyze. This also
provides a technical advantage: the stochastic-order analysis is reduced from a
bivariate random-effect problem to a univariate one, thereby simplifying the
mathematical argument.

We also note that the collective risk model with the comonotonic assumption of the latent variables is popular choice in the insurance literature; see, for example,  \citet{baumgartner2015bayesian}  and \citet{zhang2025spatially}.
 We start with the result with the identity activation functions. 

%
%
%
%
%


\begin{theorem}\label{thm.111}
Consider the setting of Model~\ref{mod.0}, except that the two-dimensional
latent vector is replaced by the comonotonic assumption
\[
R_1=R_2=:R>0 \qquad \text{a.s.},
\]
and assume that the activation functions are identity functions on the support of
the latent variable. 
Then, for 
\(
a\in(0,1],
\)
Model~\ref{mod.0} satisfies the full-claim credibility order; that is 
   \[
   \left[S_{t+1}\mid \boldsymbol{h}_{t} \right]
   \le_{\operatorname{ST}}
   \left[S_{t+1}\mid \boldsymbol{h}_{t}'\right]
   \]
   for every $t\ge 1$ and all histories $\boldsymbol{h}_{t}, \boldsymbol{h}_{t}'\in \mathcal{H}_t$ satisfying $\boldsymbol{h}_{t}\le_{\mathcal{H}} \boldsymbol{h}_{t}'$.

\end{theorem}

\begin{proof}
Fix $t\ge 1$ and take two histories
$\boldsymbol{h}_t=(n_{1:t},\boldsymbol{y}_{1:t})$ and
$\boldsymbol{h}_t'=(n_{1:t}',\boldsymbol{y}_{1:t}')$ in $\mathcal{H}_t$ satisfying
$\boldsymbol{h}_t\le_{\mathcal{H}}\boldsymbol{h}_t'$.
Let
\[
s_j:=\sum_{k=1}^{n_j}y_{j,k},
\qquad
s_j':=\sum_{k=1}^{n_j'}y_{j,k}',
\qquad j=1,\ldots,t.
\]
Since $\boldsymbol{h}_t\le_{\mathcal{H}}\boldsymbol{h}_t'$, we have 
\[
s_j=\sum_{k=1}^{n_j}y_{j,k}\ \le\ \sum_{k=1}^{n_j}y_{j,k}'\ \le\ \sum_{k=1}^{n_j'}y_{j,k}'=s_j',
\qquad j=1,\ldots,t,
\]
and hence
\begin{equation}\label{eq:St_ordered}
\sum_{j=1}^t \frac{s_j}{\mu_j}\ \le\ \sum_{j=1}^t \frac{s_j'}{\mu_j}.
\end{equation}

Then, up to multiplicative constants not depending on $R$, the joint likelihood of $\boldsymbol{h}_t$ given $R$
is
\[
\begin{aligned}
f\!\left(\boldsymbol{h}_t\mid R \right)
&\propto
\prod_{j=1}^t
\Big[e^{-\lambda_j R}(\lambda_j R)^{n_j}\Big]\,
\prod_{j=1}^t\prod_{k=1}^{n_j}\Big[(\mu_j R)^{-a}e^{-y_{j,k}/(\mu_j R)}\Big]\\
&\propto
R^{(1-a)\sum_{j=1}^t n_j}\exp\left(- R\sum_{j=1}^t \lambda_j -\frac{\sum_{j=1}^t s_j/\mu_j}{R}\right).
\end{aligned}
\]

Consequently, the ratio of the posterior densities satisfies, again up to a multiplicative constant not
depending on $R$,
\begin{equation}\label{eq:post_ratio_thm111}
\frac{\pi(R\mid \boldsymbol{h}_t')}{\pi(R\mid \boldsymbol{h}_t)}
\propto
R^{(1-a)\left(\sum_{j=1}^t (n_j'-n_j)\right)}\,
\exp\!\left(-\frac{\sum_{j=1}^t (s_j'-s_j)/\mu_j}{R}\right).
\end{equation}
By $\boldsymbol{h}_t\le_{\mathcal{H}}\boldsymbol{h}_t'$, the map
\[
r\mapsto r^{(1-a)\left(\sum_{j=1}^t (n_j'-n_j)\right)}
\]
 is nondecreasing on $(0,\infty)$ for $a\in(0,1]$, and the map
\[
r\mapsto \exp\!\left(-\frac{\sum_{j=1}^t (s_j'-s_j)/\mu_j}{r}\right)
\] is also nondecreasing on $(0,\infty)$. Hence the right-hand side of
\eqref{eq:post_ratio_thm111} is a nondecreasing function of $R>0$, which implies
\begin{equation}\label{eq.121}
\left[
R\mid \boldsymbol{h}_t'\right] 
\ge_{\mathrm{LR}}\ 
\left[
R\mid \boldsymbol{h}_t\right],
\end{equation}
and in particular $\left[R\mid \boldsymbol{h}_t'\right]\ge_{\mathrm{ST}} \left[R\mid \boldsymbol{h}_t\right]$.

Therefore, Proposition~\ref{prop.p1} along with \eqref{eq.121} and Lemma~\ref{lem.2} in Appendix yields
\[
\left[
S_{t+1}\mid \boldsymbol{h}_t\right]
\le_{\mathrm{ST}}\ 
\left[S_{t+1}\mid \boldsymbol{h}_t'\right].
\]
This completes the proof.
\end{proof}

%
%


Theorem~\ref{thm.111} establishes the full-claim credibility order under the identity activation
functions, but it requires the rather strong restriction, $a\in(0,1]$, on the shape parameter of the severity distribution. 
Alternatively, rather than giving the restriction to the shape parameter of the severity distribution, the following result shows that we can utilize
specific combination of activation functions for $\sigma_1$ and $\sigma_2$.

\begin{theorem}\label{thm.112}
Consider the setting of Model~\ref{mod.0}, except that the two-dimensional
latent vector is replaced by the assumption that 
 the
latent vector $(R_1,R_2)$ is comonotonic and the link functions satisfies
\begin{equation}\label{eq:link_sigma1_sigma2}
\sigma_1(R_1)=\left(\sigma_2(R_2)\right)^{b}\qquad \text{a.s.}
\end{equation}
for some constant $b\ge a$.
Then, Model~\ref{mod.0} satisfies the full-claim credibility order; that is 
   \[
   \left[S_{t+1}\mid \boldsymbol{h}_{t} \right]
   \le_{\operatorname{ST}}\left[S_{t+1}\mid \boldsymbol{h}_{t}'\right]
   \]
   for every $t\ge 1$ and all histories $\boldsymbol{h}_{t}, \boldsymbol{h}_{t}'\in \mathcal{H}_t$ satisfying $\boldsymbol{h}_{t}\le_{\mathcal{H}} \boldsymbol{h}_{t}'$.
   
\end{theorem}

\begin{proof}
Fix $t\ge 1$ and take two histories
$\boldsymbol{h}_t=(n_{1:t},\boldsymbol{y}_{1:t})$ and
$\boldsymbol{h}_t'=(n_{1:t}',\boldsymbol{y}_{1:t}')$ in $\mathcal{H}_t$ satisfying
$\boldsymbol{h}_t\le_{\mathcal{H}}\boldsymbol{h}_t'$.
Let $s_j:=\sum_{k=1}^{n_j}y_{j,k}$ and $s_j':=\sum_{k=1}^{n_j'}y_{j,k}'$.
Since we have $\boldsymbol{h}_t\le_{\mathcal{H}}\boldsymbol{h}_t'$, the same logic as in Theorem \ref{thm.111} induces \eqref{eq:St_ordered}.

Under the assumption \eqref{eq:link_sigma1_sigma2}, set
\[
R:=\sigma_2(R_2)\in(0,\infty),
\qquad\text{so that}\qquad
\sigma_1(R_1)=R^{b}\quad\text{a.s.}
\]
so that, with the conditional independence structure inherited from
Model~\ref{mod.0}, we have
\[
\left[N_j\mid R\right]\sim \mathrm{Poisson}(\lambda_j R^{b}),
\qquad
\left[Y_{j,k}\mid R\right]\sim \mathrm{Gamma}(a,\mu_jR).
\]
As a result, up to multiplicative constants not depending on $R$, the joint
likelihood of $\boldsymbol{h}_t$ given $R$ is
\begin{equation}\label{eq.11}
\begin{aligned}
f\left(\boldsymbol{h}_t\mid R \right)
&\propto
\prod_{j=1}^t
\Big[e^{-\lambda_j R^{b}}(\lambda_j R^{b})^{n_j}\Big]\,
\prod_{j=1}^t\prod_{k=1}^{n_j}\Big[(\mu_j R)^{-a}e^{-y_{j,k}/(\mu_j R)}\Big]\\
&\propto
\exp\!\left(-\sum\limits_{j=1}^t\lambda_j R^{b}-\sum_{j=1}^{t}\frac{s_j}{\mu_j R}\right)R^{(b-a)\sum_{j=1}^{t}n_j}.
\end{aligned}
\end{equation}

Then, the ratio of the posterior densities is
\[
\frac{\pi(R\mid \boldsymbol{h}_t')}{\pi(R\mid \boldsymbol{h}_t)}
\propto
\exp\left(-\sum\limits_{j=1}^{t}\frac{(s_j'-s_j)}{\mu_j R}\right)R^{(b-a)\sum\limits_{j=1}^{t}(n_j'-n_j)},
\]
which is clearly non-decreasing function of $R\in(0,\infty)$. Thus
\begin{equation}\label{eq:LR_order_R}
\left[R\mid \boldsymbol{h}_t'\right]\ \ge_{\mathrm{LR}}\ \left[R\mid \boldsymbol{h}_t\right],
\end{equation}
and in particular $\left[R\mid \boldsymbol{h}_t'\right]\ge_{\mathrm{ST}} \left[R\mid \boldsymbol{h}_t\right]$.
Finally, Proposition \ref{prop.p1} based on \eqref{eq:LR_order_R} and Lemma \ref{lem.3} in Appendix conclude the proof.

\end{proof}

Theorems~\ref{thm.111} and~\ref{thm.112} rely substantially on
distribution-specific likelihood structures. Nevertheless, the same proof
strategy can be extended beyond the Poisson--Gamma specification. In particular, the Gamma distributional assumption for the severity component
can be relaxed to a general scale-family specification, which covers a broad
class of severity distributions. We also show that the Poisson distributional assumption for the
frequency component can be replaced by a more flexible parametric count
distribution that contains the Poisson distribution as a special case.
Such extensions of Theorem~\ref{thm.111} are discussed in Section~\ref{sec.6}.

%
%

\section{Empirical study using LGPIF insurance data}\label{sec.5}

\subsection{Data description}

The data analyzed in this study come from the Wisconsin Local Government
Property Insurance Fund (LGPIF), which provides property coverage for various
local government entities as well as other policyholders classified in a
miscellaneous category \citep{frees2016multivariate}. The dataset covers the
period from 2006 to 2011, which we index by \(t=1,\ldots,6\). We use the
observations from 2006 to 2010, corresponding to \(t=1,\ldots,5\), as the
training set, and the observations from 2011, corresponding to \(t=6\), as the
validation set. We focus on collision coverage for new and old vehicles and
exclude policyholders for whom both coverages are zero. After these exclusions, the in-sample data consist of $k=409$ policyholders,
while the validation set contains $k'=253$ policyholders.

The model includes two categorical covariates: the type of local government
entity, with six levels, and a three-level category constructed from the
collision coverage amount. 
The response variables are claim frequency and aggregate loss; 
Table~\ref{tab:entity_summary} summarizes them by entity type, while
Table~\ref{tab:yearly_summary} reports their summary statistics over the
training period. For a more detailed description of the dataset, see
\citet{frees2016multivariate}.

\begin{table}[h!t!]
\centering
\begin{tabular}{ccccc}
\toprule
& & \multicolumn{2}{c}{Frequency($N$)} & Aggregate loss($S$) \\
\cmidrule(lr){3-4}\cmidrule(lr){5-5}
Entity type & Number of Observations & Mean & Proportion of zero & Mean \\
\midrule
City    & 47 & 1.01  & 0.543 & 4685  \\
County  & 55 & 3.43  & 0.155 & 26544 \\
Miscellaneous   & 12  & 0.25  & 0.800 & 1433  \\
School  & 161 & 0.379 & 0.761 & 1675  \\
Town    & 53 & 0.100 & 0.906 & 1197  \\
Village & 81 & 0.362 & 0.738 & 3272  \\
\midrule
Total & 409 & 0.930 & 0.645 & 6455 \\
\bottomrule
\end{tabular}
\caption{Summary statistics of covariates by entity type}
\label{tab:entity_summary}
\end{table}

\begin{table}[h!t!]
\centering
\begin{tabular}{ccccccc}
\toprule
& & \multicolumn{3}{c}{Frequency($N$)} & \multicolumn{2}{c}{Aggregate loss($S$)} \\
\cmidrule(lr){3-5}\cmidrule(lr){6-7}
Year & Obs. & Mean & Variance & Proportion of zero & Mean & Standard deviation \\
\midrule
2006 & 319 & 0.702 & 2.32 & 0.671 & 6039 & 22851 \\
2007 & 302 & 0.917 & 3.82 & 0.649 & 6344 & 18102 \\
2008 & 286 & 1.050 & 5.56 & 0.654 & 7975 & 26968 \\
2009 & 279 & 0.957 & 4.92 & 0.634 & 6087 & 17401 \\
2010 & 281 & 1.060 & 5.61 & 0.612 & 5865 & 18632 \\
\midrule 
Total & 1467 & 0.930 & 4.39 & 0.645 & 6455 & 21125 \\
\bottomrule
\end{tabular}
\caption{Summary statistics of response variables by year}
\label{tab:yearly_summary}
\end{table}
 
\subsection{Model specifications}

We consider three subclasses of Model~\ref{mod.0}. The first two are the proposed models, Models 3 and 4, which are based on comonotonic latent variables and satisfy the assumptions of Theorems~\ref{thm.111} and~\ref{thm.112}, respectively. Consequently, both models satisfy the full-claim credibility order.
\begin{itemize}
  \item \textbf{Model 3.} For
  \[
  a\in(0,1]
  \qquad\text{and}\qquad
  \sigma\in\Real_{>0},
  \]
  we consider Model~\ref{mod.0} with the following parametrization
  \[
  \left[N_t\mid R\right] \sim \mathrm{Poisson}(\lambda_t R),
  \qquad
  \left[Y_{t,j}\mid R\right] \sim \mathrm{Gamma}(a,\mu_t R),
  \]
  where $R=\operatorname{exp}(Q)$
  \[
  Q \sim \operatorname{N}\!\left(-\frac{1}{2}\sigma^2,\sigma^2\right).
  \]

  \item \textbf{Model 4.} For
  \[
  a,\sigma\in\Real_{>0}
  \qquad\text{and}\qquad
  b\in[a,\infty),
  \]
  we consider Model~\ref{mod.0} with the following parametrization
  \[
  \left[N_t\mid R\right] \sim \mathrm{Poisson}(\lambda_t R^b),
  \qquad
  \left[Y_{t,j}\mid R\right] \sim \mathrm{Gamma}(a,\mu_t R)
  \]
  where $R=\operatorname{exp}(Q)$
  \[
  Q \sim \operatorname{N}\!\left(-\frac{1}{2}\sigma^2,\sigma^2\right).
  \]
\end{itemize}
As a benchmark, we also consider Model 5, which is driven by a bivariate latent vector and does not, in general, satisfy the full-claim credibility order. This specification is commonly used in the actuarial literature \citep{jeong2020predictive, oh2021predictive}.
\begin{itemize}
  \item \textbf{Model 5.} For
  \[
  a\in\Real_{>0}, \qquad \rho\in(-1,1)
  \qquad\text{and}\qquad
  \sigma_1,\sigma_2\in\Real_{>0},
  \]
  we consider Model~\ref{mod.0} with the following parametrization
  \[
  \left[N_t\mid R_1,R_2\right] \sim \mathrm{Poisson}(\lambda_t R_1),
  \qquad
  \left[Y_{t,j}\mid R_1,R_2\right] \sim \mathrm{Gamma}(a,\mu_t R_2),
  \]
where
\[
(R_1,R_2)=\left(\exp(Q_1),\exp(Q_2)\right),
\]
and \((Q_1,Q_2)\) follows a bivariate normal distribution with mean vector and
covariance matrix given by
\[
  \begin{pmatrix}
  -\frac{1}{2}\sigma_1^2\\[0.3ex]
  -\frac{1}{2}\sigma_2^2
  \end{pmatrix}
  \quad\hbox{and}\quad 
  \begin{pmatrix}
  \sigma_1^2 & \rho\,\sigma_1\sigma_2 \\
  \rho\,\sigma_1\sigma_2 & \sigma_2^2
  \end{pmatrix},
\]
 respectively.
  
\end{itemize}

Model~\ref{mod.0} is formulated for a single policyholder. To apply it to panel data, we introduce an additional index \(i=1,\ldots,k\) to represent the \(i\)-th policyholder and assume that the corresponding random sequences are independent across \(i\). For simplicity, we assume that the \emph{a priori} rates are constant over time whose values are fixed at time $t=1$, that is,
\[
\lambda_{i,t}=\lambda_{i,1}
\qquad\text{and}\qquad
\mu_{i,t}=\mu_{i,1},
\qquad \forall t\in\mathbb{N},
\]
and set
\[
\lambda_i\equiv\lambda_{i,1}=\operatorname{exp}\left(\langle \mathbf{x}_i,\boldsymbol w_1\rangle\right),
\qquad
\mu_i\equiv\mu_i=\operatorname{exp}\left(\langle \mathbf{x}_i,\boldsymbol w_2\rangle\right),
\]
where \(\boldsymbol w_1,\boldsymbol w_2\in\Real^{p+1}\) are regression coefficients and \(\mathbf{x}_i\in\Real^{p+1}\) is a time-invariant covariate vector for the \(i\)-th policyholder, including an intercept term.

Inference under these models is generally nontrivial because the likelihood
function involves integration over the latent random effects and is typically
not available in closed form. For the proposed comonotonic models, this
integration is only one-dimensional and can therefore be handled accurately by
standard numerical quadrature \citep{breslow1993approximate}. However, a likelihood-based implementation would still require repeated
numerical integration during optimization and would need to be tailored to the structure of each random-effect specification. Alternatively, we
adopt a \emph{Markov chain Monte Carlo (MCMC)} approach, which provides a unified
computational framework for all model specifications considered in this study.
The MCMC framework also directly yields posterior predictive distributions and
uncertainty quantification, which are central to our model assessment. We
implement the MCMC algorithm in \texttt{NIMBLE}, which allows custom
Metropolis--Hastings moves and blocked updates for both the latent effects and
the model parameters \citep{de2017programming, nimble}.

\subsection{Model assessment}

Out-of-sample predictive performance is assessed on the 2011 validation set
using the \emph{predictive mean squared error} (PMSE) and the \emph{predictive mean absolute error}
(predictive PMAE). These measures are computed separately for claim frequency and aggregate
loss. For a response \(y_{i,6}\), representing either \(N_{i,6}\) or \(S_{i,6}\),
let \(\hat y_{i,6}\) denote the corresponding predictive mean. With
\(k'=253\) policyholders in the validation set, we define
\[
\operatorname{PMSE}
=
\frac{1}{k'}\sum_{i=1}^{k'}
\left(y_{i,6}-\hat y_{i,6}\right)^2,
\qquad
\operatorname{PMAE}
=
\frac{1}{k'}\sum_{i=1}^{k'}
\left|y_{i,6}-\hat y_{i,6}\right|.
\]

We also evaluate the \emph{predictive log-likelihood} (PLL) on the 2011
validation set. For claim frequency, the PLL is defined as
\[
\mathrm{PLL}_{N}
=
\sum_{i=1}^{k}
\log \widehat p_N\!\left(n_{i,6}\mid n_{i,5}\right),
\]
where \(\widehat p_N(\cdot\mid n_{i,5})\) denotes the posterior
predictive probability mass function of \(N_{i,6}\). For aggregate loss, the
PLL is defined as
\[
\mathrm{PLL}_{S}
=
\sum_{i=1}^{k}
\log \widehat f_S\!\left(s_{i,6}\mid \boldsymbol h_{i,5}\right),
\]
where \(\widehat f_S(\cdot\mid \boldsymbol h_{i,5})\) denotes the corresponding
predictive density, or probability mass when \(s_{i,6}=0\). Larger
values of PLL indicate better out-of-sample predictive performance. All
quantities are estimated from the MCMC output.

In addition, we report the \emph{Watanabe--Akaike information criterion} (WAIC), computed on the 2006--2010 training data from the MCMC output. WAIC is an in-sample Bayesian model-comparison criterion that evaluates predictive fit while adjusting for effective model complexity. Smaller values of WAIC indicate a better trade-off between fit and complexity.

Table~\ref{tab.2} summarizes the results for predicting frequency \(N_{i,6}\) and aggregate loss \(S_{i,6}\). On the validation set, Model 5 achieved the smallest PMSE and PMAE for both frequency and aggregate loss, indicating the best overall predictive accuracy among the three models. Model 5 also attained the largest PLL and the smallest WAIC, suggesting the best out-of-sample predictive performance as well as the best in-sample predictive fit after accounting for model complexity. Model 4 ranked second across all criteria. In particular, its frequency prediction performance, as measured by PMSE, PMAE, and PLL, was close to that of Model 5, although the gap was slightly larger for aggregate loss. 
Model 3 showed the weakest performance among the three models, although its
frequency prediction performance remained reasonably competitive.

\begin{table}[h!t!]
\centering
\begin{tabular}{cccccccc}
\toprule
& \multicolumn{3}{c}{frequency} & \multicolumn{3}{c}{aggregate loss} & \multirow{2}{*}{WAIC}\\
\cmidrule(lr){2-4}\cmidrule(lr){5-7}
& MSE & MAE & PLL & MSE & MAE & PLL & \\
\midrule
Model 3 & 1.141 & 0.657 & -220.121 & 464960200 & 7807.873 & -994.8281 & 14022.92 \\
Model 4 & 1.094 & 0.607 & -209.365 & 310684700 & 6758.867 & -986.7147 & 13678.08 \\
Model 5 & 1.081 & 0.597 & -207.853 & 259304300 & 6198.940 & -976.5138 & 13577.32 \\
\bottomrule
\end{tabular}
\caption{Out-of-sample predictive performance on the 2011 validation set and WAIC on the 2006--2010 training set.}\label{tab.2}
\end{table}

Despite its favorable predictive performance in Table~\ref{tab.2}, we further assess Model 5 by examining whether it respects the full-claim credibility order under a worsening of the observed claim history. For each policyholder \(i=1,\ldots,k\) in the training set, let
\[
\boldsymbol{h}_{i,5}=(n_{i,1:5},\,\boldsymbol{y}_{i,1:5})
\]
denote the full-claim history up to time \(5\). For a fixed \(\epsilon=0.1\), we define a perturbed counterfactual history by adding one additional claim of size \(\epsilon\) in the most recent period:
\[
\boldsymbol{h}_{i,5}(\epsilon):=(n_{i,1:5}^*,\,\boldsymbol{y}_{i,1:5}^*),
\]
where
\[
n_{i,t}^*=n_{i,t},
\qquad
\boldsymbol{y}_{i,t}^*=\boldsymbol{y}_{i,t},
\qquad t=1,\ldots,4,
\]
and
\[
n_{i,5}^*=n_{i,5}+1,
\qquad
\boldsymbol{y}_{i,5}^*=\begin{cases}
\bigl(y_{i,5,1},\ldots,y_{i,5,n_{i,5}},\epsilon\bigr), & n_{i,5}>0;\\
(\epsilon), &n_{i,5}=0.\\
\end{cases}
\]
That is, \(\boldsymbol{h}_{i,5}(\epsilon)\) is obtained from \(\boldsymbol{h}_{i,5}\) by worsening the most recent claim history through one additional claim with severity \(\epsilon\).

For various values of deductible \(d\ge 0\), Table~\ref{tab.3} reports the number and percentage of violations of the monotonicity condition
\begin{equation}\label{eq.41}
\E{\left(S_{i,6}-d\right)_+\mid \boldsymbol{h}_{i,5}}
>
\E{\left(S_{i,6}-d\right)_+\mid \boldsymbol{h}_{i,5}(\epsilon)},
\qquad i=1,\ldots,k
\end{equation}
which implies the violation of the full-claim credibility order.
Ratio of violations at \(d=0\) is about $2.20\%$, and their frequency increases as the deductible grows. This indicates that the premium-reversal phenomenon persists under deductible insurance and may become more pronounced for larger deductibles in our example. All premiums in \eqref{eq.41} are estimated using \(600{,}000\) Monte Carlo simulations of the aggregate loss \(S_{i,6}\).

By contrast, Models~3 and~4 satisfy the full-claim credibility order by
construction. For Model~3, the estimate $\widehat{a}=0.41\in(0,1]$ satisfies
the condition of Theorem~\ref{thm.111}, thereby guaranteeing the full-claim
credibility order. Model~4 is automatically guaranteed by
Theorem~\ref{thm.112} to satisfy the full-claim credibility order. In
particular, Model~4 achieves predictive performance close to that of Model~5
while avoiding violations of the full-claim credibility order. Therefore, the proposed models provide
an alternative to the classical benchmark model, retaining competitive
predictive performance while guaranteeing the full-claim credibility order.

\begin{table}[h!t!]
\centering
\begin{tabular}{cccccc}\toprule
  &\multicolumn{5}{c}{Deductible-applied} \\
\cmidrule(lr){2-6}
& $d$=0& $d$=0.5$\bar{S}$& $d$=$\bar{S}$& $d$=1.5$\bar{S}$& $d$=2$\bar{S}$\\\midrule

percentage&  2.20\% & 3.18\%& 6.11\%& 10.27\%&  13.94\%\\
count &  9& 13& 25& 42&  57\\
\bottomrule
\end{tabular}
\caption{Counts and percentages of policyholders for whom Model~5 violates the full-claim credibility order under deductible premiums. The deductible levels are expressed in terms of \(\bar S\), where \(\bar S=6455\) is the sample mean aggregate loss in the training data reported in Table~\ref{tab:entity_summary}.}
\label{tab.3}
\end{table}

\section{Extension to various distributions}\label{sec.6}

This section discusses extensions of the full-claim credibility order established in
Theorem~\ref{thm.111} to alternative distributional specifications for the
severity and frequency components of the collective risk model.
Since the proof of Theorem~\ref{thm.111} relies on the specific likelihood structure
of the Poisson--Gamma model, a fully general extension of both components is not
available in this framework. Nevertheless, the same proof strategy can be adapted
to several useful settings.

We first consider the severity component. Specifically, we replace the Gamma
severity model by a general scale-family specification generated from an arbitrary
baseline severity distribution, and derive a monotonicity condition under which
the full-claim credibility order is preserved. This provides a broad extension
for the severity part. We then turn to the frequency component,
where the extension is more limited: we replace the Poisson frequency model by
the Conway--Maxwell--Poisson distribution, which provides a practically relevant
but distribution-specific generalization of the frequency model.

\subsection{Severity part}

First, we replace the Gamma severity assumption in Model~\ref{mod.0} with a scale-family specification generated by an arbitrary baseline severity distribution.
Let $Y^*$ be a continuous random variable supported on $(0,\infty)$ with density
\begin{equation}\label{scale1}
g(y), \qquad y>0.
\end{equation}
For a given parameter $\theta>0$, consider the scaled random variable $Y=\theta Y^*$.
Then the density of $Y$ is
\begin{equation}\label{scale}
g^{\rm [scale]}(y\,;\,\theta):=\frac{1}{\theta}g\!\left(\frac{y}{\theta}\right),\qquad y>0.
\end{equation}
Under this replacement of the severity model, we obtain the following extension of
Theorem~\ref{thm.111}.

\begin{theorem}\label{thm.212}
Consider the setting of Model~\ref{mod.0}, except that the two-dimensional
latent vector is replaced by the assumption that the activation functions are identity functions,
\[
\sigma_1(x)=\sigma_2(x)=x,\qquad x\in\Real,
\]
and that the latent vector $(R_1,R_2)$ is comonotonic in the following sense 
\[
R_1=R_2=:R>0 \qquad\text{a.s.}
\]
For a given density function $g$ in \eqref{scale1}, 
replace the severity assumption in part~ii of Model~\ref{mod.0} by 
\begin{itemize}
%
\item[ii.] \textbf{Severity:}
Conditional on \(R_1=R_2=R\), the severities
\(\{Y_{t,j}\}_{t\ge1,\ j\ge1}\) are mutually independent, with conditional density
\begin{equation}\label{eq.t1}
f(y_{t,j}\mid R)
:=
g^{\rm [scale]}(y_{t,j};\mu_t R),
\qquad t\ge1,\ j\ge1.
\end{equation}
Moreover, conditional on \(R\), the severity array
\(\{Y_{t,j}\}_{t\ge1,\ j\ge1}\) is independent of the claim-count process
\(\{N_t\}_{t\ge1}\), where \(g^{\rm [scale]}\) is defined in \eqref{scale}.
\end{itemize}

Assume the following assumption:
\begin{itemize}
  \item [{\bf A1.}] \(g:(0,\infty)\to(0,\infty)\) is differentiable and that
\[
y\mapsto -y\frac{\partial}{\partial y}\log g(y)
\]
is nonnegative and nondecreasing on \((0,\infty)\).
\end{itemize}

Then, the above collective risk model satisfies the full-claim credibility order; that is,
\[
\left[S_{t+1}\mid \boldsymbol{h}_{t}\right] \le_{\operatorname{ST}} \left[S_{t+1}\mid \boldsymbol{h}_{t}'\right]
\]
for every $t\ge 1$ and all histories $\boldsymbol{h}_{t}, \boldsymbol{h}_{t}'\in \mathcal{H}_t$
satisfying $\boldsymbol{h}_{t}\le_{\mathcal{H}} \boldsymbol{h}_{t}'$.
\end{theorem}

\begin{proof}
Fix $t\ge 1$ and take two histories
$\boldsymbol{h}_t=(n_{1:t},\boldsymbol{y}_{1:t})$ and
$\boldsymbol{h}_t'=(n_{1:t}',\boldsymbol{y}_{1:t}')$ in $\mathcal{H}_t$ satisfying
$\boldsymbol{h}_t\le_{\mathcal{H}}\boldsymbol{h}_t'$.
Under the comonotonicity assumption $R_1=R_2=:R$, conditional on $R$ the model becomes
\[
\left[N_j\mid R \right]\sim \mathrm{Poisson}(\lambda_j R),
\]
and, conditional on \(R\), the variables \(Y_{j,k}\) are mutually independent with density
\[
f(y\mid R)
=
\frac{1}{\mu_j R}
g\!\left(\frac{y}{\mu_j R}\right).
\]
Hence, up to multiplicative constants not depending on $R$,
the joint likelihood of $\boldsymbol{h}_t$ given $R$ is
\[
\begin{aligned}
f\!\left(\boldsymbol{h}_t\mid R\right)
&\propto
\prod_{j=1}^t \Big[e^{-\lambda_j R}(\lambda_j R)^{n_j}\Big]\,
\prod_{j=1}^t\prod_{k=1}^{n_j}\Big[(\mu_j R)^{-1}\,g\!\left(\frac{y_{j,k}}{\mu_j R}\right)\Big]\\
&\propto
\prod_{j=1}^t e^{-\lambda_j R}\,
\prod_{j=1}^t\prod_{k=1}^{n_j} g\!\left(\frac{y_{j,k}}{\mu_j R}\right),
\end{aligned}
\]
up to a multiplicative constant not depending on $R$.

Consequently, the ratio of the posterior densities satisfies, again up to a multiplicative constant not
depending on $R$,
\begin{equation}\label{eq:post_ratio_thm212}
\frac{\pi(r\mid \boldsymbol{h}_t')}{\pi(r\mid \boldsymbol{h}_t)}
\propto
\prod_{j=1}^t\prod_{k=1}^{n_j}
\frac{g\!\left(\frac{y_{j,k}'}{\mu_j r}\right)}{g\!\left(\frac{y_{j,k}}{\mu_j r}\right)}
\ \cdot\
\prod_{j=1}^t\prod_{k=n_j+1}^{n_j'} g\!\left(\frac{y_{j,k}'}{\mu_j r}\right).
\end{equation}

Then, we have
\[
\begin{aligned}
  \frac{\partial}{\partial r} \log\frac{\pi(r\mid \boldsymbol{h}_t')}{\pi(r\mid \boldsymbol{h}_t)}
  &=\frac{\partial}{\partial r}
  \left[
  \sum_{j=1}^t\sum_{k=1}^{n_j}\left(\log g\!\left(\frac{y_{j,k}'}{\mu_j r}\right) 
  -\log g\!\left(\frac{y_{j,k}}{\mu_j r}\right)
  \right)
  \right] +
  \frac{\partial}{\partial r}
  \left[
  \sum_{j=1}^t\sum_{k=n_j+1}^{n_j'}\log g\!\left(\frac{y_{j,k}'}{\mu_j r}\right)
  \right]\\
  &=\frac{1}{r}\left[
  \sum_{j=1}^t\sum_{k=1}^{n_j}\left(
  -x\frac{\partial}{\partial x} \log g(x)\bigg\vert_{x=\frac{y_{j,k}'}{\mu_j r}}
  -\left(-x\frac{\partial}{\partial x} \log g(x)\bigg\vert_{x=\frac{y_{j,k}}{\mu_j r}}\right)
  \right)
  \right]\\
  &\qquad\qquad\qquad\qquad\qquad\qquad\qquad
  + \frac{1}{r}\left[
  \sum_{j=1}^t\sum_{k=n_j+1}^{n_j'}\left( -x\frac{\partial}{\partial x} \log g(x)\bigg\vert_{x=\frac{y_{j,k}'}{\mu_j r}}\right)
  \right]\\
  &\ge 0
\end{aligned}
\]
where the first term in the second line is non-negative due to the assumption that 
\[
-y\frac{\partial}{\partial y}\log g(y)
\]
is non-decreasing function of $y>0$, and 
the second term in the second line is non-negative due to the assumption 
$-y\frac{\partial}{\partial y}\log g(y)\ge 0$.
As a result, we have
\[
r\mapsto \frac{\pi(r\mid \boldsymbol{h}_t')}{\pi(r\mid \boldsymbol{h}_t)}
\]
is non-decreasing function on $(0, \infty)$, which implies
\[
\left[R\mid \boldsymbol{h}_t'\right]\ge_{\mathrm{LR}}\left[R\mid \boldsymbol{h}_t\right]
\]
and in particular 
\[
\left[R\mid \boldsymbol{h}_t'\right]\ge_{\mathrm{ST}}\left[R\mid \boldsymbol{h}_t\right]
\]

Therefore, this along with Proposition~\ref{prop.p1}  and
Lemma~\ref{lem.212} in \ref{app.2} yields
\[
\left[S_{t+1}\mid \boldsymbol{h}_t\right]\ \le_{\mathrm{ST}}\ \left[S_{t+1}\mid \boldsymbol{h}_t'\right].
\]
This completes the proof.

\end{proof}

The key sufficient condition for the full-claim credibility order is given in
A1. In the Gamma severity specification of
Theorem~\ref{thm.111}, the condition in A1 holds if and only if
\(a\in(0,1]\), which recovers the restriction imposed in
Theorem~\ref{thm.111}.
The following examples provide severity density functions \(g\) in
\eqref{scale1}, beyond the Gamma distribution, that satisfy A1
 in Theorem~\ref{thm.212}.

\begin{example}\label{ex:weibull_scale_family}
Consider the Weibull density with shape parameter $\alpha>0$ and unit scale parameter $\beta=1$:
\[
g(y\,;\,\alpha)=
\begin{cases}
\alpha\,y^{\alpha-1}e^{-y^{\alpha}}, & y>0,\\
0, & y\le 0.
\end{cases}
\]
Then, for $y>0$,
\[
\log g(y\,;\,\alpha)=\log \alpha +(\alpha-1)\log y-y^{\alpha},
\]
and hence
\[
h(y\,;\,\alpha):=
- y \frac{\partial}{\partial y}\log g(y\,;\,\alpha)
=1-\alpha+\alpha y^{\alpha}.
\]
Moreover,
\[
\frac{\partial}{\partial y}h(y\,;\,\alpha)
=\alpha^{2}y^{\alpha-1}\ge 0,
\qquad y>0.
\]
Therefore, $h(y\,;\,\alpha)$ is nondecreasing on $(0,\infty)$. In addition, if $\alpha\le 1$, then
\[
h(y\,;\,\alpha)=1-\alpha+\alpha y^{\alpha}\ge 1-\alpha\ge 0,
\qquad y>0.
\]
Thus, when $\alpha\le 1$, the condition in Theorem~\ref{thm.212} is satisfied. In this case,
\eqref{eq.t1} corresponds to the specification
\[
\left[Y_{t,j}\mid R\right] \sim {\rm Weibull}(\alpha,\mu_t R)
\]
\end{example}

\begin{example}\label{ex:transformed_gamma}
Consider the transformed gamma density with the first shape parameter $\alpha>0$, the second shape parameter $\tau>0$, and unit scale parameter $\beta=1$:
\[
g(y\,;\,\alpha,\tau)=
\begin{cases}
\displaystyle
\frac{\tau}{\Gamma(\alpha)}\,y^{\alpha\tau-1}e^{-y^{\tau}}, & y>0,\\[1.2ex]
0, & y\le 0.
\end{cases}
\]
Then, for $y>0$,
\[
\log g(y\,;\,\alpha,\tau)
=\log \tau-\log \Gamma(\alpha)+(\alpha\tau-1)\log y-y^{\tau},
\]
and hence
\[
h(y\,;\,\alpha,\tau):=
- y \frac{\partial}{\partial y}\log g(y\,;\,\alpha,\tau)
=1-\alpha\tau+\tau y^{\tau}.
\]
Moreover,
\[
\frac{\partial}{\partial y}h(y\,;\,\alpha,\tau)
=\tau^{2}y^{\tau-1}\ge 0,
\qquad y>0.
\]
Therefore, $h(y\,;\,\alpha,\tau)$ is nondecreasing on $(0,\infty)$. In addition, if $\alpha\tau\le 1$, then
\[
h(y\,;\,\alpha,\tau)=1-\alpha\tau+\tau y^{\tau}\ge 1-\alpha\tau\ge 0,
\qquad y>0.
\]
Thus, when $\alpha\tau\le 1$, the condition in Theorem~\ref{thm.212} is satisfied. In this case,
\eqref{eq.t1} corresponds to the specification
\[
\left[Y_{t,j}\mid R\right] \sim {\rm Transformed\text{-}Gamma}(\alpha,\tau,\mu_t R).
\]
\end{example}

\subsection{Frequency part}

We then replace/extend the Poisson frequency assumption in
Model~\ref{mod.0} with the Conway--Maxwell--Poisson (CMP) distribution
\citep{conway1961queueing}. The CMP family introduces an additional dispersion
parameter and nests the Poisson distribution as the special case $\nu=1$; it is
therefore useful when claim counts exhibit over- or under-dispersion relative to
Poisson, and it has been adopted in wide range of applications including insurance; see, for
example, \citet{shmueli2005useful, lord2008application, yin2024flexible}.
Formally, for parameters $\lambda>0$ and $\nu>0$, we say that
$X\sim \operatorname{CMP}(\lambda,\nu)$ if it has probability mass function
\[
\P{X=x\mid \lambda,\nu}
=\frac{1}{C(\lambda,\nu)}\frac{\lambda^x}{(x!)^\nu},
\qquad x\in\mathbb{N}_{0},
\]
where $C(\lambda,\nu)$ is the normalizing constant
\[
C(\lambda,\nu):=\sum_{j=0}^{\infty}\frac{\lambda^j}{(j!)^\nu}.
\]
Under this replacement of the frequency model, we obtain the following analogue of
Theorem~\ref{thm.111}.

\begin{theorem}\label{thm.211}
Consider the setting of Model~\ref{mod.0}, where the frequency assumption in part~i is replaced by the Conway--Maxwell--Poisson distribution: for some fixed $\nu>0$,
\begin{itemize}
\item[i.] \textbf{Frequency:} $\left[N_t \mid R_1, R_2\right] \sim \mathrm{CMP}(\lambda_t \sigma_1(R_1), \nu)$ independently over $t$.
\end{itemize}
Assumption of the two-dimensional
latent vector is replaced by the comonotonic assumption
\[
R_1=R_2=:R>0 \qquad \text{a.s.},
\]
and assume that the activation functions are identity functions on the support of
the latent variable. 
Then, for $a\in(0,1]$, the above collective risk model satisfies the full-claim credibility order; that is,
\[
\left[S_{t+1}\mid \boldsymbol{h}_{t}\right] \le_{\operatorname{ST}} 
\left[S_{t+1}\mid \boldsymbol{h}_{t}'\right]
\]
for every $t\ge 1$ and all histories $\boldsymbol{h}_{t}, \boldsymbol{h}_{t}'\in \mathcal{H}_t$
satisfying $\boldsymbol{h}_{t}\le_{\mathcal{H}} \boldsymbol{h}_{t}'$.
\end{theorem}

\begin{proof}
Fix $t\ge 1$ and take two histories
$\boldsymbol{h}_t=(n_{1:t},\boldsymbol{y}_{1:t})$ and
$\boldsymbol{h}_t'=(n_{1:t}',\boldsymbol{y}_{1:t}')$ in $\mathcal{H}_t$ satisfying
$\boldsymbol{h}_t\le_{\mathcal{H}}\boldsymbol{h}_t'$.
Let
\[
s_j:=\sum_{k=1}^{n_j}y_{j,k},
\qquad
s_j':=\sum_{k=1}^{n_j'}y_{j,k}',
\qquad j=1,\ldots,t.
\]
Since $\boldsymbol{h}_t\le_{\mathcal{H}}\boldsymbol{h}_t'$, we have
\[
s_j=\sum_{k=1}^{n_j}y_{j,k}\ \le\ \sum_{k=1}^{n_j}y_{j,k}'\ \le\ \sum_{k=1}^{n_j'}y_{j,k}'=s_j',
\qquad j=1,\ldots,t,
\]
and hence
\[
\sum_{j=1}^t \frac{s_j}{\mu_j}\ \le\ \sum_{j=1}^t \frac{s_j'}{\mu_j}.
\]

Under the present model, conditional on $R=r>0$, we have
\[
\left[N_j\mid r\right] \sim \mathrm{CMP}(\lambda_j r,\nu),
\qquad
\left[Y_{j,k}\mid r\right] \sim \mathrm{Gamma}(a,\text{scale}=\mu_j r),
\]
with the severity variables mutually independent over \(j\) and \(k\), conditional on \(R=r\), and conditionally independent of the claim counts.
Then, up to multiplicative constants not depending on $R=r$, the joint likelihood of
$\boldsymbol{h}_t$ given $R$ is
\[
\begin{aligned}
f\!\left(\boldsymbol{h}_t\mid r \right)
&\propto
\prod_{j=1}^t
\left[
\frac{1}{C(\lambda_j r,\nu)}\frac{(\lambda_j r)^{n_j}}{(n_j!)^\nu}
\right]
\prod_{j=1}^t\prod_{k=1}^{n_j}
\left[
(\mu_j r)^{-a}\exp\!\left(-\frac{y_{j,k}}{\mu_j r}\right)
\right]\\
&\propto
\left(\prod_{j=1}^t \frac{1}{C(\lambda_j r,\nu)}\right)
r^{(1-a)\sum_{j=1}^t n_j}\,
\exp\!\left(-\frac{\sum_{j=1}^t s_j/\mu_j}{r}\right),
\end{aligned}
\]
where $C(\cdot,\nu)$ is the CMP normalizing constant.

Consequently, the ratio of the posterior densities satisfies, again up to a multiplicative constant not
depending on $R=r$,
\begin{equation}\label{eq:post_ratio_thm211}
\frac{\pi(r\mid \boldsymbol{h}_t')}{\pi(r\mid \boldsymbol{h}_t)}
\propto
r^{(1-a)\left(\sum_{j=1}^t (n_j'-n_j)\right)}\,
\exp\!\left(-a\frac{\sum_{j=1}^t (s_j'-s_j)/\mu_j}{r}\right),
\end{equation}
where the factors $\prod_{j=1}^t C(\lambda_j r,\nu)^{-1}$ cancel out in the ratio because they do not
depend on the realized counts $n_j$.

By $\boldsymbol{h}_t\le_{\mathcal{H}}\boldsymbol{h}_t'$, we have
$\sum_{j=1}^t (n_j'-n_j)\ge 0$ and $\sum_{j=1}^t (s_j'-s_j)/\mu_j\ge 0$.
Hence, for $a\in(0,1]$, the map
\[
r\mapsto r^{(1-a)\left(\sum_{j=1}^t (n_j'-n_j)\right)}
\]
is nondecreasing on $(0,\infty)$, and the map
\[
r\mapsto \exp\!\left(-\frac{\sum_{j=1}^t (s_j'-s_j)/\mu_j}{r}\right)
\]
is also nondecreasing on $(0,\infty)$. Therefore, the right-hand side of
\eqref{eq:post_ratio_thm211} is nondecreasing in $R=r>0$, which implies
\begin{equation}\label{eq:LR_order_thm211}
\left[R\mid \boldsymbol{h}_t'\right]\ \ge_{\mathrm{LR}}\ \left[R\mid \boldsymbol{h}_t\right],
\end{equation}
and in particular $\left[R\mid \boldsymbol{h}_t'\right]\ge_{\mathrm{ST}} \left[R\mid \boldsymbol{h}_t\right]$.

%

Therefore, this along with Proposition~\ref{prop.p1} and Lemma \ref{lem.200} in \ref{app.2} yields
\[
\left[S_{t+1}\mid \boldsymbol{h}_t\right]\ \le_{\mathrm{ST}}\ \left[S_{t+1}\mid \boldsymbol{h}_t'\right].
\]
This completes the proof.
\end{proof}

Theorems~\ref{thm.212} and~\ref{thm.211} extend Theorem~\ref{thm.111} in two
separate directions. The former replaces the Gamma severity assumption with the
scale-family specification in~\eqref{scale}, whereas the latter replaces the
Poisson frequency assumption with the CMP distribution. Both extensions are
formulated under the identity activation functions of Theorem~\ref{thm.111}.
These two extensions can also be combined, yielding a version of
Theorem~\ref{thm.111} with a CMP frequency component and a scale-family severity
component. Analogous variants under the link-function setting of
Theorem~\ref{thm.112} can likewise be considered; however, to keep the present
paper focused, we do not pursue these additional variants here.

\section{Conclusion}

This paper develops a credibility-order framework tailored to collective risk
models and identifies tractable conditions under which a worsening full-claim
history leads to a stochastically larger predictive distribution for future
aggregate loss. This monotonicity property is practically important in insurance
ratemaking: if it is violated, worse claim experience may lead to a lower
premium. Such premium reversals create undesirable reporting incentives and may
undermine the integrity of experience rating. Our empirical analysis further
shows that this phenomenon can occur at a non-negligible frequency in standard
collective risk models with bivariate random effects, a common benchmark
specification in the actuarial literature \citep{jeong2020predictive,
oh2021predictive}.

To address this issue, we introduced a new notion of full-claim credibility
order and derived sufficient conditions under which it holds. For clarity, we
focused on static random-effect models, although the same ideas may be extended
to dynamic random-effect settings under suitable assumptions on the latent
process. The current theory remains partly distribution-specific: while the
severity component admits a broad scale-family extension through
Theorem~\ref{thm.212}, the frequency-side results are more restrictive and are
currently limited to the Poisson and CMP settings considered in this paper.
Extending the framework to richer frequency models and to dynamic random-effect
settings therefore remains a promising direction for future research.

%
%

\section*{Data Availability}
The data and code used in this study are available at:\\
\url{https://github.com/jaeyoun-ahn/CRM-Order}


\section*{Acknowledgements}

This research was supported by the National Research Foundation of Korea (NRF) funded by the Ministry of Education (2021R1A6A1A10039823) and by the Ministry of Science and ICT (RS-2025-23524530). Parts of this work were completed while Jihyun Park and Jieun Kim were visiting the University of Regina, and while Jae Youn Ahn was staying at Macquarie University during his sabbatical leave as a Visiting Professor.

\pagebreak
\bibliographystyle{apalike}
\bibliography{bib_tex}

\pagebreak
\appendix

\section{Auxiliary results}\label{app.2}

The following are the auxiliary results for Theorem \ref{thm.111}.
\begin{lemma}\label{lem.2}
Consider the setting in Theorem~\ref{thm.111}. Fix $0<r<r'$. Then
\[
\left[S_{t+1}\mid R=r\right]\ \le_{\mathrm{ST}}\ \left[S_{t+1}\mid R=r'\right].
\]
\end{lemma}

\begin{proof}
Fix $0<r<r'$. Let $(N,N^*,\{Z_k\}_{k\ge1})$ be independent random variables such that
\[
N\sim \mathrm{Poisson}(\lambda_{t+1} r),\qquad
N^*\sim \mathrm{Poisson}\bigl(\lambda_{t+1}(r'-r)\bigr),\qquad
Z_k\ \iid\ \mathrm{Gamma}(a,\text{scale}=1).
\]
Define
\[
S:=\mu_{t+1}r\sum_{k=1}^{N} Z_k,
\qquad
S':=\mu_{t+1}r'\sum_{k=1}^{N+N^*} Z_k.
\]
Then, with probability $1$,
\[
\begin{aligned}
S
&=\mu_{t+1}r\sum_{k=1}^{N}Z_k\\
&\le \mu_{t+1}r\sum_{k=1}^{N+N^*}Z_k\\
&\le \mu_{t+1}r'\sum_{k=1}^{N+N^*}Z_k
=S',
\end{aligned}
\]
where the first inequality uses $N^*\ge 0$ and $Z_k\ge 0$ a.s., and the second uses $0<r<r'$. 
Hence, since $S \le S'$ almost surely, the usual stochastic order $S \le_{\mathrm{ST}} S'$ holds; see, for example, Theorem~1.A.1 in \citet{shaked2007stochastic}.

Finally, from the following observations
\[
\mu_{t+1}rZ_k\sim \mathrm{Gamma}(a,\text{scale}=\mu_{t+1}r)\quad\hbox{and}\quad
N+N^*\sim \mathrm{Poisson}(\lambda_{t+1}r'), 
\]
we have
\[
S\ \eqd\ \left[S_{t+1}\mid R=r\right]
\qquad\text{and}\qquad
S'\ \eqd\ \left[S_{t+1}\mid R=r'\right].
\]
This completes the proof.
\end{proof}

The following result is the version of Lemma \ref{lem.2} in Appendix with the comonotonic assumption of the latent variables satisfying \eqref{eq:link_sigma1_sigma2}. This result is used in the proofs of Theorem \ref{thm.112}.

\begin{lemma}\label{lem.3}
Consider the setting in Theorem~\ref{thm.112}. Fix $0< r<r'$.
Then
\[
\left[S_{t+1}\mid R=r\right]\ \le_{\mathrm{ST}}\ \left[S_{t+1}\mid R=r'\right].
\]
\end{lemma}
\begin{proof}
Fix $0< r<r'$. Let $(N,N^*,\{Z_k\}_{k\ge1})$ be independent random variables such that
\[
N\sim \mathrm{Poisson}(\lambda_{t+1} r^{b}),\qquad
N^*\sim \mathrm{Poisson}\bigl(\lambda_{t+1}((r')^{b}-r^{b})\bigr),\qquad
Z_k\ \iid\ \mathrm{Gamma}(a,1).
\]
Define
\[
S:=\mu_{t+1}r\sum_{k=1}^{N} Z_k,
\qquad
S':=\mu_{t+1}r'\sum_{k=1}^{N+N^*} Z_k.
\]

Then, with probability $1$, we have the following 
\[ 
\begin{aligned} 
S &=\mu_{t+1}r\sum_{k=1}^{N}Z_k\\ 
&\le \mu_{t+1}r\sum_{k=1}^{N+N^*}Z_k\\ 
&\le \mu_{t+1}r'\sum_{k=1}^{N+N^*}Z_k=S' 
\end{aligned} 
\]
where the first inequality is from the observations that $N^*\ge 0$ and $Z_k\ge 0$ with probability $1$.
Hence, since $S \le S'$ almost surely, the usual stochastic order $S \le_{\mathrm{ST}} S'$ holds again by Theorem~1.A.1 in \citet{shaked2007stochastic}.

Finally, from $\mu_{t+1}rZ_k\sim \mathrm{Gamma}(a,\mu_{t+1}r)$ and
$\mu_{t+1}r'Z_k\sim \mathrm{Gamma}(a,\mu_{t+1}r')$, and 
$N+N^*\sim \mathrm{Poisson}(\lambda_{t+1} (r')^b)$, we conclude 
\[
S\ \eqd\ \left[S_{t+1}\mid R=r\right]
\qquad\text{and}\qquad
S'\ \eqd\ \left[S_{t+1}\mid R=r'\right].
\]
This completes the proof.
\end{proof}

The following lemma is the auxiliary result of Theorem \ref{thm.211}.

\begin{lemma}\label{lem.200}
Consider the setting in Theorem~\ref{thm.211}. Fix $0<r<r'$. Then
\[
\left[S_{t+1}\mid R=r\right]\ \le_{\mathrm{ST}}\ \left[S_{t+1}\mid R=r'\right].
\]
\end{lemma}
\begin{proof}
  Fix $0<r<r'$. For the CMP pmf,
\[
\P{N=n\mid R=r}
=
\frac{1}{C(\lambda_{t+1} r,\nu)}\frac{(\lambda_{t+1} r)^n}{(n!)^\nu},
\qquad n\in\mathbb{N}_{0},
\]
we have, for each $n\ge 0$,
\[
\frac{\P{N=n\mid R=r'}}{\P{N=n\mid R=r}}
=
\frac{C(\lambda_{t+1} r,\nu)}{C(\lambda_{t+1} r',\nu)}
\left(\frac{r'}{r}\right)^n,
\]
which is increasing in $n$. Thus,
\[
\left[N_{t+1}\mid R=r\right]\ \le_{\mathrm{LR}}\ \left[N_{t+1}\mid R=r'\right],
\]
and in particular $N_{t+1}\mid R=r\le_{\mathrm{ST}} N_{t+1}\mid R=r'$.
Hence, one can construct a coupling $(\widetilde N,\widetilde N')$ such that
\[
\widetilde N \ \eqd\ \left[N_{t+1}\mid R=r\right],
\qquad
\widetilde N' \ \eqd\ \left[N_{t+1}\mid R=r'\right],
\qquad
\widetilde N\le \widetilde N'\ \ \text{a.s.}
\]
Let $\{Z_k\}_{k\ge 1}$ be i.i.d.\ $\mathrm{Gamma}(a,1)$, independent of $(\widetilde N,\widetilde N')$,
and define
\[
S:=\mu_{t+1}r\sum_{k=1}^{\widetilde N} Z_k,
\qquad
S':=\mu_{t+1}r'\sum_{k=1}^{\widetilde N'} Z_k.
\]
Then, with probability $1$,
\[
S
=\mu_{t+1}r\sum_{k=1}^{\widetilde N} Z_k
\ \le\
\mu_{t+1}r\sum_{k=1}^{\widetilde N'} Z_k
\ \le\
\mu_{t+1}r'\sum_{k=1}^{\widetilde N'} Z_k
=S',
\]
and therefore $S\le_{\mathrm{ST}} S'$. 
\end{proof}

The following lemma is the auxiliary result of Theorem \ref{thm.212}.
This lemma is the version of Lemma \ref{lem.2} where Gamma distributional assumption for the severity part is replaced by arbitrary distribution
\[
g^{\rm [scale]}(x\,;\, \theta)
\]
defined in \eqref{scale} as in the setting in Theorem~\ref{thm.212}. Since the proof is similar to that of Lemma \ref{lem.2}, we omit the proof here. 
\begin{lemma}\label{lem.212}
Consider the setting in Theorem~\ref{thm.212}. Fix $0<r<r'$. Then
\[
\left[S_{t+1}\mid R=r\right]\ \le_{\mathrm{ST}}\ \left[S_{t+1}\mid R=r'\right].
\]
\end{lemma}

\end{document}